\RequirePackage{fix-cm}
\documentclass[11pt]{article}

\usepackage{setspace}
 \usepackage{tikz}
 \usepackage{pgfplots}
 \pgfplotsset{compat=1.18}
 \usepackage{subcaption}

\usepackage[T1]{fontenc}
\usepackage[latin9]{inputenc}
\usepackage{amsmath,amssymb,amsthm,mathtools}
\usepackage{enumitem}
\usepackage{bbm}
\usepackage{hyphenat}
\usepackage{hyperref}
\usepackage[authoryear]{natbib}
\makeatletter

\usepackage[a4paper]{geometry}
\PassOptionsToPackage{normalem}{ulem}
\usepackage{ulem}
\usepackage{microtype}

\numberwithin{equation}{section}
\theoremstyle{plain}

\numberwithin{equation}{section}

\theoremstyle{plain}
\newtheorem{theorem}{Theorem}[section]
\newtheorem{proposition}[theorem]{Proposition}
\newtheorem{lemma}[theorem]{Lemma}
\newtheorem{corollary}[theorem]{Corollary}

\theoremstyle{definition}

\theoremstyle{remark}

\theoremstyle{definition}
\newtheorem{example}{Example}[section]

\newcommand{\R}{\mathbb{R}}
\newcommand{\E}{\mathbb{E}}
\newcommand{\Prob}{\mathbb{P}}
\newcommand{\cav}{\operatorname{cav}}
\newcommand{\supp}{\operatorname{supp}}
\newcommand{\1}{\mathbf{1}}

\newcommand{\join}{\vee}
\newcommand{\meet}{\wedge}
\newcommand{\SSO}{\succeq_{\mathrm{SSO}}}
\newcommand{\cx}{\succeq_{cx}}

\title{\textbf{Bayesian Persuasion under Bias Management}}
\author{Kemal Ozbek\thanks{Department of Economics, University of Southampton, University Road, Southampton, S017 1BJ, United Kingdom. Email: mkemalozbek@gmail.com}}
\date{\today}

\begin{document}
\maketitle

\begin{abstract}
A principal delegates choice to an agent whose decision depends on both beliefs and tastes.
The principal can steer the delegated decision using two costly instruments:
(i) an information policy that determines a Bayes--plausible distribution of posteriors,
and (ii) a bias-management policy that shifts the agent's effective taste.
We study a binary-state, two-action, convex hull of two benchmark tastes specialization with posterior-separable information costs.
The analysis admits an inner--outer decomposition:
optimal bias management is bang--bang (either no intervention or the minimal intervention needed to flip the agent's action),
while the optimal information policy is characterized by concavification of an endogenous posterior value function that already incorporates optimal management and information costs.
This structure clarifies how information acquisition and bias management interact; they can be complements, substitutes, or both depending on the primitives of the model.
Information changes which posteriors are realized and hence where management is used;
management reshapes the curvature and kinks of the posterior value function and hence the marginal value of information.
The model delivers regime classifications for pooling vs. informativeness and for management at different posteriors within informative signals,
and highlights how comparative statics can be monotone or non-monotone depending on how concavification contact points move with costs.
\end{abstract}

\bigskip
\noindent\textbf{Keywords:} delegation, information design, Bayesian persuasion, bias management, rational inattention, concavification.\\
\textbf{JEL codes:} D82, D83, D86.

\section{Introduction}

Delegation is a central organizing principle in economics: public agencies delegate vendor selection and contract administration to procurement officers; boards delegate corporate capital allocation to managers; retailers delegate pricing and consumer finance offers to frontline sales teams; regulators delegate compliance and remediation decisions to firms; and banks delegate credit underwriting to loan officers and scoring teams while trying to steer those choices toward broader social or organizational objectives.

Delegation is valuable because agents often hold local knowledge and can react to contingencies,
but it creates a classic control problem: the agent's decision rule reflects not only information but also \emph{bias}---a divergence between the agent's objective and the principal's. Two influential literatures isolate different steering channels. Delegation theory studies how a principal steers a biased agent by restricting discretion---committing to a set of allowable actions from which the agent later chooses after observing information. Separately, Bayesian persuasion and information design study how a sender influences a receiver by committing to an information structure that shapes posterior beliefs, with optimal policies characterized via concavification.

In many real environments, principals can also engage in a third activity:
\emph{bias management}---training, culture, screening, internal controls, governance procedures, compliance capacity,
and other interventions that change how the agent maps beliefs into actions.
Bias management is especially relevant when transfers are limited or non-contractible, or when bias arises from organizational culture, bounded rationality, or mission preferences.
This motivates a joint design question:
How should a principal optimally combine information acquisition and bias management to steer delegated choice?
How do these two costly activities interact---are they complements or substitutes, and can their relationship be non-monotone?

This paper proposes a tractable joint framework that integrates persuasion-style information choice with a reduced-form technology for bias management.
The principal chooses at a cost (i) a Bayes--plausible information policy over posteriors and (ii) after each posterior is realized,
a management action that shifts the agent's effective taste within the convex hull of two benchmark tastes.
The convex-hull specification captures \emph{partial alignment}: the principal cannot rewrite the agent's objective, but can move it systematically. This structure makes the interdependence between the two activities transparent:
information determines which posteriors occur (and thus where management is used),
while management reshapes the curvature and kinks of the posterior value function (and thus the marginal value of information).
As a result, information and bias management can be complements or substitutes,
and comparative statics can exhibit jumps when concavification contact points change.

\subsection{A joint design problem}
\label{subsec:expanded_discussion}

We specialize to the binary-state $S=\{H,L\}$, two-action $\{f,g\}$, convex-hull management $q\in[0,1]$ case that yields sharp results. The central mechanism in the model is an inner--outer decomposition that separates \emph{local} (posterior-by-posterior) bias management from \emph{global} (Bayes--plausible) information design. This decomposition clarifies which forces operate at which stage of the principal's problem and why information acquisition and bias management are jointly determined. The inner problem delivers \emph{bang--bang} management: at each posterior, the principal either does nothing or applies the minimal management needed to induce the preferred action.
The outer problem remains a persuasion-style concavification, but applied to an endogenous posterior value function that already incorporates the management option and nets out posterior-separable information costs.

At a realized posterior \(p\), the principal's management decision affects only one object: whether the agent's best response crosses the action threshold (from choosing \(f\) to choosing \(g\)).
Because the agent's choice is binary and the management technology shifts the effective taste within a convex hull, the principal faces a discrete tradeoff:
either do not intervene and accept the agent's default action, or intervene just enough to induce the desired action.
Any ``over-management'' beyond the point at which the agent is induced to choose \(g\) is strictly wasteful, because it does not change the action but increases cost.
This yields the bang--bang structure: \(q^*(p;k_V)\in\{0,q_{\min}(p)\}\).
Economically, this captures a common feature of governance and compliance interventions:
they are deployed not to fine-tune behavior, but to clear decision thresholds (e.g., to ensure a project is rejected unless evidence is strong, or to ensure compliance once indicators exceed a trigger).

Once the optimal management rule is substituted in, the principal's outer problem becomes the choice of a Bayes--plausible distribution $\tau \in \Delta([0,1])$ of posteriors to maximize an \emph{endogenous} posterior value function net of information costs.
In classical Bayesian persuasion, the shape of the sender's value function is pinned down by primitives and the receiver's best response.
Here, the value of a posterior includes an option value of management.
As such, management reshapes the curvature of the posterior value function: it can create or remove kinks at posteriors where the management decision switches between \(0\) and \(q_{\min}(p)\).
Because the concavification geometry depends on these kinks, the optimal information structure can change sharply when parameters vary.

A distinctive implication of the joint design problem is that informativeness need not move monotonically with the management cost \(k_V\).
Holding the posterior fixed, increasing \(k_V\) weakly reduces the value of the management option, which tends to reduce the gains from splitting beliefs.
However, changing \(k_V\) also shifts the region where management is used and can introduce or remove kinks in the posterior value function.
These shape changes can alter the supporting hyperplane at the prior and hence the concavification contact points.
As a result, the optimal pair of posteriors \((p_-,p_+)\) can move discontinuously, producing jumps in informativeness even when the underlying primitives vary smoothly.
In economic terms, a small change in the feasibility or cost of alignment can cause the principal to retarget information acquisition toward a different set of posteriors where management is most effective, thereby increasing or decreasing informativeness depending on the direction of the retargeting.

Information and management become \emph{substitutes} when management is sufficiently cheap and robust that it can correct behavior across a wide range of posteriors.
In that case, the principal can ``buy'' correct decisions directly via management, reducing the incremental value of generating finely dispersed posteriors.
Coarser information---including pooling at the prior---can become optimal because additional informativeness only marginally improves outcomes once behavior is already aligned.
This logic parallels settings in which strong governance or strict protocols reduce the need for extensive monitoring and reporting, and conversely, where detailed reporting is most useful precisely when governance is weak.

Information and management become \emph{complements} when management makes information more actionable.
For example, if the agent's default action is often misaligned in an intermediate range of beliefs, then producing posteriors concentrated in that range is of limited value unless the principal can also induce the correct action there.
In such cases, investing in management raises the marginal payoff of informative signals: it increases the difference in payoffs between high and low posteriors (by ensuring that high-confidence posteriors translate into the principal-preferred action), thereby increasing the concavification gap at the prior and making informativeness more attractive.
In organizational terms, improving alignment (through training, culture, governance) can make data and analytics more valuable because decisions respond appropriately to the information.

As such, the framework suggests that policy evaluations of steering tools should account not only for \emph{substitution} across levers but also for \emph{complementarity} between them. Restrictions on choice architecture (limits on nudges, defaults, or imposed frictions) may shift optimal design toward more intensive information provision, but they can also change \emph{where} posteriors concentrate, thereby altering the marginal value of governance or incentive reforms that shape decision rules at those beliefs. Likewise, stronger disclosure mandates may reduce the returns to costly managerial or regulatory oversight when information alone pushes beliefs into regions where choices are aligned, yet the same mandates may \emph{increase} the returns to oversight when they generate more frequent ``knife-edge'' or intermediate posteriors where alignment is pivotal and management determines which action is taken. More generally, steering is best viewed as a portfolio choice across belief-based and rule-based interventions: depending on costs and on how information reallocates probability mass across posteriors, the optimal portfolio can feature substitution (one lever replacing the other) or complementarity (information making management more valuable, or management increasing the payoff to information).

\subsection{Related literature}

Bias management can be viewed as a reduced form for the principal's ability to shape an agent's incentives and internal governance when decisions are delegated.
In classic delegation models, the principal trades off \emph{discretion} and \emph{control} by restricting the agent's action set, anticipating that an informed but biased agent will choose actions that reflect his private objectives (e.g., \cite{holmstrom1977incentives, alonsoMatouschek2008,amadorBagwell2013}).
Our setting keeps that delegation logic---an agent with his own ``taste'' chooses from a menu---but adds two margins that are typically held fixed:
the principal can make the agent \emph{better informed} (choosing an information policy $\tau$) as in Bayesian persuasion and can also \emph{reduce the effective bias} through costly governance or incentive interventions (choosing a management policy $q$) which is novel to our model.
In this sense, the framework endogenizes both the agent's information and the severity of his bias, allowing us to study when the principal substitutes between investing in information and investing in organizational controls, and how these levers interact through which posteriors lead the delegated agent to take the principal-preferred action.

\citet{kamenicaGentzkow2011} show that a committed sender's problem can be written as choosing a
Bayes--plausible distribution of posteriors and, in the binary-state case, solved by concavifying a
reduced-form posterior payoff. We keep this outer geometry---$\tau$ is still Bayes--plausible over
$p$---but depart from standard persuasion because the posterior payoff is not primitive: at each
posterior the principal also chooses a costly bias-management/implementation instrument (here
$q\in[0,1]$ in the convex-hull specialization) that reshapes the agent's decision rule and thus the
value $\phi(p)$ being concavified. \cite{kolotilinZapechelynyuk2025} show that, under standard
assumptions, delegation and persuasion are equivalent; because $\phi$ is endogenously determined
here, our two-layer problem fits neither literature directly. Instead it sits between Bayesian
persuasion and broader information-design frameworks (\citet{bergemannMorris2019}): in the
obedience/Bayes-correlated formulation the principal chooses an information device subject to
incentive constraints, while in our setting management shifts the effective constraint set by moving the agent's underlying ``type'' (taste), changing which actions are incentive compatible at a given posterior.

We model information acquisition costs using posterior-separable and convex penalties, a standard reduced-form way to capture limited attention and costly information processing in rational inattention and related frameworks (\cite{sims2003,caplinDean2015,matejkaMckay2015}).
In our binary-state specialization, this takes the form $c_P(\tau)=k_P\int \kappa(p)\,\tau(dp)$ with convex $\kappa$, so the outer problem remains a concavification of a net posterior payoff.
Relative to the rational inattention literature, the novelty is not the cost functional but the \emph{endogeneity} of the posterior payoff being concavified:
at each posterior the principal can also choose a costly bias-management action ($q\in[0,1]$) that reshapes the agent's response to beliefs.
This yields a sharp interaction: the marginal value of information is jointly determined by $k_P$ and $k_V$.
When bias management is expensive (high $k_V$), the principal may optimally demand \emph{decisive} signals---those that move posteriors into regions where tastes agree and no intervention is needed---whereas when bias management is cheap, the principal can rely on local steering at intermediate posteriors and may pool more under the same information cost.

A large literature on nudges, defaults, and framing emphasizes that policy and platform interventions can operate either by changing agents' \emph{beliefs} (an information channel) or by changing their \emph{decision rule holding beliefs fixed} (a behavioral or ``preference'' channel), as in the influential synthesis of \citet{thalerSunstein2008}.
Related work on shrouded attributes and obfuscation highlights how market institutions can strategically reduce the informativeness of the environment, thereby shaping behavior through the information channel (\cite{gabaixLaibson2006}).
Empirical surveys of persuasion and media effects similarly distinguish belief-based from preference-based mechanisms and document that both can matter in practice (\cite{dellaVignaGentzkow2010}).
Our framework provides a tractable unifying representation of these two channels in a delegation setting:
the principal chooses an information policy $\tau$ that determines the distribution of posteriors, and also chooses a costly bias-management instrument $q(.)$ that shifts the agent's cutoff rule at a given posterior.
In this reduced form, nudges and defaults correspond to interventions that move the posterior-to-action mapping, while informational disclosures or suppression correspond to changes in $\tau$; the distinctive prediction is that the optimal use of one lever depends on the cost and effectiveness of the other.

Our menu-of-acts formulation connects the analysis to the literature on flexibility, commitment, and endogenous choice frictions.
In the classic ``desire for flexibility'' account of \citet{Kreps1979}, a decision maker values larger opportunity sets because she anticipates future taste uncertainty; \citet{DekelLipmanRustichini2001} provide axioms and a representation in which menu preferences are generated by a subjective state space of future utilities.
We use menus in a different role: the menu $A$ is the \emph{delegation instrument} chosen by the principal, and taste variation reflects heterogeneity or bias on the part of the delegated agent.
A distinctive feature of our framework is that the principal can partially \emph{design} the agent's effective taste---via a costly bias-management technology (captured by $q(.)$)---so that the trade-off between flexibility and control interacts with investments in information.
This perspective is closely related to work on \emph{costly information acquisition} within menu-choice environments (\cite{dDMO17}), where a decision maker endogenously chooses how much to learn before selecting from a menu, and to work on \emph{costly self-regulation} (\cite{MihmOzbek2018}), where costly effort reshapes the decision rule applied to a given menu.
Recent axiomatization formalizes sequential variants of these interactions: \cite{HOT25} propose and characterize a ``sequential persuasion'' model in which a principal chooses information policies and subsequent constraints in a staged fashion, building on the menu-choice techniques developed in the flexibility, commitment, and costly choice-friction literatures.
Our two-layer problem can be viewed as a tractable analogue of these themes: information design determines the distribution of posteriors faced at the choice stage, while bias management determines how the delegated chooser maps a posterior into an action within the given menu.

The rest of the paper proceeds as follows.
Section~\ref{sec:model} presents the framework: primitives, timing, costs, and solution concept; and also interprets the primitives in the context of motivating examples.
Section~\ref{sec:analysis} provides the analysis: we characterize optimal management at each posterior,
solve the information problem via concavification, analyze monotone informativeness in the two-posterior domain, summarize regime structure and trajectories across cost parameters, and discuss regulatory implications of the two--layer model. Section~\ref{sec:conclusion} concludes. Numerical examples and proofs are given in an Appendix.

\section{Framework}\label{sec:model}

We present primitives and then specialize to the binary-state, two-action, convex-hull management case that yields sharp results.

Let \(X\) be a finite set of outcomes and \(\Delta(X)\) lotteries.
Let \(S\) be a finite set of states and \(P=\Delta(S)\) beliefs.
Let \(F\) be the set of acts \(f:S\to\Delta(X)\).
A menu \(A\subset F\) is a nonempty closed set of acts.
The principal's taste is \(u\in \mathbb{R}^{|X|}\).
The agent's taste is \(v\in \mathbb{R}^{|X|}\).
Given menu \(A\) and belief \(p\), the agent selects $\arg\max_{f\in m_{v,p}(A)}\, u(f)\cdot p$ where $m_{v,p}(A) = \arg\max_{f\in A} v(f)\cdot p.$ As such, the agent chooses a best act within the set of available options and breaks any possible ties in favor of the principal.

We specialize to \(S=\{H,L\}\), identify beliefs with \(p=\Prob(H)\in[0,1]\), and fix prior \(p_0\in(0,1)\). Thus, with slight abuse of notation, we use $p$ both for a belief and the probability weight the belief $p$ assigns to the state $H$.
An information policy is \(\tau\in\Delta([0,1])\) satisfying Bayes plausibility:
\(
\int_0^1 p\,\tau(dp)=p_0.
\) We denote the set of Bayes plausible information policies by $\Delta_0([0,1])$. Let the menu be \(A=\{f,g\}\).
Define the principal's posterior gain from implementing \(g\) rather than \(f\) by
\(
\Delta_u(p):=u(g)\cdot p-u(f)\cdot p.
\)
We assume that \(\Delta_u(\cdot)\) is increasing on \([0,1]\).

The agent's effective taste is chosen within the convex hull of two benchmark tastes \(v_L\) and \(v_H\) such that \(v_q=(1-q)v_H+qv_L\) for all \(q \in [0,1].\)
We assume that there exist \(0<\pi_L<\pi_H<1\) and a continuous strictly decreasing function \(\pi:[0,1]\to[\pi_L,\pi_H]\) such that an agent with taste \(v_q\) chooses \(g\) at belief \(p\) iff \(p\ge \pi(q)\).
Moreover \(\pi(0)=\pi_H\) and \(\pi(1)=\pi_L\). Define the minimal management needed to induce \(g\) at posterior \(p\in[\pi_L,\pi_H]\) as
\(
q_{\min}(p):=\inf\{q\in[0,1]:\ p\ge \pi(q)\}.
\)
Then \(q_{\min}(p)=1\) for \(p= \pi_L\), \(q_{\min}(p)=0\) for \(p= \pi_H\),
and \(q_{\min}(p)\in(0,1)\) for \(p\in(\pi_L,\pi_H)\). Moreover \(q_{\min}(\cdot)\) is decreasing.

Information cost is posterior-separable:
\(
c_P(\tau)=k_P\int_0^1 \kappa(p)\,\tau(dp),
\)
where \(k_P\ge 0\) and \(\kappa(.)\) is continuous and convex. Management cost is convex:
\(
c_V(q)=k_V C(q),
\)
where \(k_V\ge 0\) and \(C:[0,1]\to\R_+\) is increasing, l.s.c., convex, and \(C(0)=0\). Given an information-management policy $(\tau,q(\cdot))$, the expected payoff is
\[
\int^1_{0}[\,\max_{h\in m_{v_{q(p)},p}(A)}u(h).p\;-k_V\,C(q(p))\,]\,\tau(dp)\;-\,k_P\,\int^1_{0}\,\kappa(p)\,\tau(dp).
\]

\emph{Timing.}
(i) The principal commits to an information policy \(\tau\) and an after-posterior management rule \(q(\cdot)\);
(ii) A posterior \(p\) is drawn according to \(\tau\);
(iii) The principal chooses \(q(p)\) and pays \(k_V\,C(q(p))\);
(iv) The agent observes \(p\) and chooses \(h\in\{f,g\}\) to maximize \(v_{q(p)}(h)\cdot p\);
(v) The principal receives \(u(h)\cdot p\) from the chosen act, minus costs.

\emph{Optimal policy.} An optimal policy $(\tau^*,q^*(\cdot))$ maximizes expected payoff subject to Bayes plausibility and the agent's best response.

\subsection{Information before management vs management before information}\label{subsec:timing}

Our baseline model assumes that the principal designs an information policy and then, after
each posterior is realized, chooses a bias-management action. An
alternative timing reverses these choices: the principal commits to a single management
choice ex ante and subsequently designs/acquires information. We compare below the
principal's payoffs under the two timings and shows that the baseline timing 
dominates.

In the convex-hull specialization, bias management is a choice of intensity $q\in[0,1]$ that
selects the agent's taste $v_q=qv_L+(1-q)v_H$ (equivalently, a cutoff $\pi(q)$). Let
$
B(p,q):=\max_{h\in m_{v_q,p}(A)} u(h)\cdot p
$
denote the principal's payoff at posterior $p$ when management level $q$ is used, and let
the management cost be $c_V(q)$. The information policy is a Bayes--plausible distribution
$\tau$ over posteriors with cost $c_P(\tau)$.

\smallskip
The two-layer problem corresponding to \emph{baseline timing} is 
\[
U^{\mathrm{bas}}
=\max_{\tau}\left\{\int \Big[\max_{q\in[0,1]}\big(B(p,q)-c_V(q)\big)\Big]\tau(dp)-c_P(\tau)\right\},
\]

\smallskip
and the two-layer problem corresponding to \emph{reversed timing} is
\[
U^{\mathrm{rev}}
=\max_{q\in[0,1]}\left\{\max_{\tau}\Big[\int B(p,q)\,\tau(dp)-c_P(\tau)\Big]-c_V(q)\right\}.
\]

\begin{proposition}\label{prop:timing}
The baseline timing dominates the
reversed timing:
$
U^{\mathrm{bas}}\ \ge\ U^{\mathrm{rev}}.
$
\end{proposition}

The inequality reflects a pure ``flexibility'' effect: in the baseline timing, choosing
$q$ after observing the posterior allows the principal to tailor management to the realized
belief, effectively spending management only where it matters (typically inside the conflict
region $(\pi_L,\pi_H)$). The reversed timing forces a single ex ante $q$ to be used at all posteriors,
creating a compromise that is generally suboptimal when the optimal $q^*(p)$ varies across
posteriors. The comparison is typically strict whenever $\tau^*$ is informative and the
baseline optimal management $q^*(p)$ is nonconstant on the support of $\tau^*$; it becomes
an equality under pooling or whenever the same management level is optimal at all posteriors
that occur.

\subsection{Interpreting the specialization in applications}
\label{subsec:app_mapping}

We briefly outline environments where a principal naturally chooses both information acquisition and bias management, and illustrate how the binary-state, two-action, convex-hull specialization
can be interpreted within these environments.

\emph{Public procurement and anticorruption: intelligence vs professionalization.}
A ministry delegates procurement to an agency that may be biased by corruption or incompetence.
The ministry can acquire information via market intelligence and benchmarking, and can manage bias via training, rotation, peer review, and approval protocols.
Intelligence without professionalization can be wasted; professionalization can raise the marginal value of intelligence by ensuring it is used appropriately. Let $H$ denote ``procurement market is distorted / high corruption risk'' and $L$ denote ``competitive/clean market.''
Action $g$ is ``use a strict protocol'' (e.g.\ open tender, enhanced due diligence), while $f$ is ``use a discretionary/fast-track procedure.''
Market intelligence (benchmarking, supplier screening) produces posterior $p$.
A biased agency may prefer discretionary procedures even at moderately high $p$ (due to convenience, relationships, or rents),
corresponding to a high cutoff.
Professionalization---training, rotation, peer review---reduces that bias, lowering the cutoff, hence increasing the probability of selecting $g$
for a given posterior.
Here, $q$ is the intensity of professionalization and oversight, and $\tau$ is the information structure from intelligence gathering.

\emph{Corporate capital allocation: analytics vs governance.}
Headquarters delegates its investment decisions to business units.
It can finance market analytics and forecasting, and manage bias by strengthening governance (investment committees, veto rights, internal controls).
Governance is particularly valuable at borderline posteriors---precisely where more informative analytics can concentrate decision weight. Let $H$ denote ``project is truly high NPV'' and $L$ denote ``project is low NPV.''
Action $g$ is ``approve/invest'' and $f$ is ``reject/defer.''
Analytics and market research generate posterior $p=\Pr(H)$.
A division manager may be biased toward toward $f$ (career risk aversion).
Bias management $q$ captures governance (investment committees, veto rights, internal controls) that shifts the manager's objective
toward the firm's benchmark, changing the cutoff $\pi(q)$ at which investment is approved.
Convex-hull management corresponds to partial alignment: governance cannot perfectly eliminate career concerns, but can reduce their weight.

\emph{Sales and consumer finance: lead scoring vs suitability culture.}
A firm delegates targeting and product recommendations to sales.
It can invest in data and lead scoring (information) and manage bias via training, incentives, and compliance culture.
Better information may amplify mis-selling unless bias is managed; alignment can make fine-grained information more valuable. Let $H$ denote ``customer is unsuitable for the product'' and $L$ denote ``customer is suitable.''
Action $g$ is ``offer safe alternative'' and $f$ is ``sell/allocate product.''
Lead scoring and underwriting produce posterior $p=\Pr(H)$.
Sales incentives may bias the agent toward $f$, raising the cutoff too much not to sell (overselling).
Bias management $q$ captures compliance culture, suitability training, and internal monitoring that lower the cutoff toward the principal's benchmark.
The binary-action framework captures the extensive-margin decision, while $p$ captures informativeness of customer signals.

\emph{Regulation and compliance: audits vs enforcement capacity.}
A regulator delegates compliance decisions to a firm.
The regulator can acquire information by mandating audits and reports, and can manage bias by investing in enforcement capacity, compliance frameworks, and governance rules.
Better audits matter only to the extent the regulator can induce action conditional on what is learned; conversely, enforcement capacity is most valuable when targeted to decisive posteriors generated by information. Let $H$ denote ``high harm'' and $L$ denote ``low harm''.
The delegated action $g$ is ``remediate'', while $f$ is ``no remediation.''
A realized audit or inspection produces posterior $p=\Pr(H)$.
The firm (agent) is biased toward $f$ (cost-saving), so absent management it requires a high posterior to take $g$:
this corresponds to a high cutoff $\pi_H$ under taste $v_H$.
Bias management $q$ captures strengthening enforcement, monitoring capacity, or governance mandates that make the firm
more responsive to evidence, lowering the cutoff to $\pi(q)$.
The convex-hull interpretation is that institutional reforms move the firm's decision rule from a ``profit-first'' taste toward a more compliant taste.
Information $\tau$ is the audit regime: how often and how decisively inspections generate high posteriors.

\emph{Credit underwriting: scoring models vs risk culture.}
A bank delegates credit decisions to loan officers.
It can invest in credit scoring and data (information) and manage bias via risk culture, training, and approval layers.
The interaction determines whether resources are better spent on more data or better governance. Let $H$ denote ``borrower is high default risk'' and $L$ denote ``low risk.''
Action $g$ is ``deny/tighten terms'' and $f$ is ``approve/standard terms.''
Credit scoring produces posterior $p$.
Loan officers may be biased toward approval (relationship lending).
Risk culture, approval layers, and audit processes shift the cutoff $\pi(q)$.
Information acquisition $\tau$ is the choice of model/data investment that affects the posterior distribution of risk assessments.

\section{Analysis}\label{sec:analysis}
In this section, we characterize optimal bias management, solve the information-design problem via concavification, study monotone informativeness on the two-posterior domain, summarize the resulting regime structure and parameter trajectories, and discuss regulatory implications.

\subsection{Solving the two-layer problem}
In the two-layer problem, the outer instrument $\tau$ (information-acquisition) controls the dispersion of posteriors around the prior $p_0$,
while the inner instrument $q$  (bias-management) reshapes the posterior-by-posterior mapping from beliefs to actions. We begin our analysis with solving the inner problem and then turn to the outer  problem.

\subsubsection{Optimal bias management at each posterior}
Given posterior \(p\) and management intensity \(q\), the agent chooses \(g\) iff \(p\ge \pi(q)\).
Hence the principal's posterior payoff is $\Psi(p,q)=u(f)\cdot p+\1\{p\ge \pi(q)\}\Delta_u(p)-k_VC(q).$ Define the posterior value function $\phi_{k_V}(p):=\max_{q\in[0,1]} \Psi(p,q).$

\begin{lemma}\label{thm:bangbang}
For every posterior \(p\in[0,1]\) there exists an optimal management choice \(q^*(p;k_V)\) such that: $q^*(p;k_V)=0 \;\text{for } p<\pi_L\ \text{and for } p\ge \pi_H,$ and for \(p\in[\pi_L,\pi_H),\)
\(
q^*(p;k_V)= q_{\min}(p) \; \text{if} \; \Delta_u(p)>k_VC(q_{\min}(p)) \; \text{and} \; q^*(p;k_V)= 0 \; \text{if} \; \Delta_u(p)\le k_VC(q_{\min}(p)),
\)
and so,
$
\phi_{k_V}(p)=u(f)\cdot p+\max\Big\{0,\ \Delta_u(p)-k_VC(q_{\min}(p))\Big\}.
$
\end{lemma}

A key implication of this result is that the optimal management intensity at a given posterior is at a corner: either no management when tastes already select the desired action, or the minimal/most cost-effective adjustment needed to flip the choice in the conflict region. As such, it is inefficient to ``spread'' small interventions broadly; instead the principal should concentrate effort on a small set of pivotal beliefs (via information) and apply management only where it changes the action, yielding sharp regime thresholds and clear comparative statics for when steering operates through information, through management, or through both.

\emph{Break-even cost.} In the convex-hull specialization, for each posterior $p\in(\pi_L,\pi_H)$ there is a
\emph{minimal} management intensity $q_{\min}(p)$ that just flips the agent's choice to the
principal-preferred action $g$. The gross gain from flipping is $\Delta_u(p)$, while the
corresponding minimal cost is $k_V C(q_{\min}(p))$. This motivates the
\emph{break-even cost} which is defined for \(p\in(\pi_L,\pi_H)\) as
\(
\bar k_V(p):=\frac{\Delta_u(p)}{C(q_{\min}(p))},
\)
with \(\bar k_V(p)=+\infty\) if \(C(q_{\min}(p))=0\). The break-even cost is the largest management-cost parameter under which it is still profitable to apply the minimal
intervention at posterior $p$. Equivalently, in the minimal-intervention regime,
$q^*(p)>0$ if and only if $k_V\le \bar k_V(p)$. \medskip

If $\Delta_u(p)$ is increasing, $q_{\min}(p)$ is decreasing, and $C$ is increasing, then $C(q_{\min}(p))$ is decreasing and hence $\bar k_V(p)$ is increasing in $p$. 

\begin{proposition}\label{lem:kbarMon}
The threshold function \( \bar k_V(p)\) is increasing on \((\pi_L,\pi_H)\) (with the \(+\infty\) convention when \(C(q_{\min}(p))=0\)).
\end{proposition}

Economically, higher posteriors make management more cost-effective:
the benefit from selecting $g$ is larger and/or the intervention needed to flip the choice is smaller, so the principal is willing to pay for management at a wider range of $k_V$.

\emph{Cutoff posterior.}
When $\bar k_V(p)$ is increasing, the set of posteriors where minimal flipping is profitable,
$\{p\in(\pi_L,\pi_H): k_V\le \bar k_V(p)\}$, is an upper interval. We summarize it by the
cutoff posterior
\[
\hat p(k_V)\;:=\;\inf\big\{p\in(\pi_L,\pi_H): k_V\le \bar k_V(p)\big\},
\]
the lowest belief at which the principal is still willing to pay for the minimal management
that induces $g$. Since $\bar k_V(p)$ is increasing, $\hat p(k_V)$ is increasing in $k_V$, so the
posterior region in which management is used shrinks monotonically as $k_V$ rises. 

\begin{corollary}\label{cor:phat}
For \(p\in(\pi_L,\pi_H)\), $\;q^*(p;k_V)>0\;$ iff $\; p\ge \hat p(k_V),$
and \(\hat p(k_V)\) is increasing.
\end{corollary}

This cutoff description is useful because shifts in $\hat p(k_V)$ change the shape of the net posterior payoff $g_{k_P,k_V}(p):=\phi_{k_V}(p)-k_P\kappa(p)$, and hence can trigger concavification contact-set reselection in the outer information design problem.

\subsubsection{Optimal information: concavification and two-posteriors}

Consider the information problem the principal solves
\[U(k_P,k_V)
=\max_{\tau\in\Delta([0,1])}\left\{\int_0^1 g_{k_P,k_V}(p)\,\tau(dp)\right\}
\quad\text{s.t.}\quad \int_0^1 p\,\tau(dp)=p_0.\]

Our outer information-design problem inherits the geometry of Bayesian persuasion. In \cite{kamenicaGentzkow2011}, a principal
commits to an information structure, equivalently to a Bayes--plausible distribution over posteriors, and in binary-state
environments the optimum is characterized by concavifying the posterior payoff function, which implies that an optimal signal
can be taken to have at most two posteriors in its support. \cite{GentzkowKamenica2014} extend this insight to \emph{costly}
persuasion by subtracting a posterior-separable convex information cost from the payoff, so that the optimal experiment is
still obtained by concavifying an \emph{effective} payoff $u(p)-k_P\kappa(p)$. Our two-layer delegation problem
preserves this outer concavification logic, but the key difference is that the payoff from each posterior is \emph{endogenous}:
at every posterior the principal can additionally choose a costly bias-management policy that reshapes the agent's decision
rule. Thus the object being concavified is $g_{k_V,k_P}(p)=\phi_{k_V}(p)-k_P\kappa(p)$, where $\phi_{k_V}(p)$ is itself the
value of an inner optimization.

We write $\cav g_{k_P,k_V}$ for the \emph{concave envelope} (smallest concave majorant) of $g_{k_P,k_V}$ on
$[0,1]$, also called the \emph{concavification} of $g_{k_P,k_V}$ in the Bayesian persuasion literature
(\cite{rockafellar1970,kamenicaGentzkow2011}). We obtain the following characterization on the optimal solution of the information acquisition problem.

\begin{lemma}\label{thm:concav}

The following are true:
\begin{enumerate}[leftmargin=18pt]
\item \(U(k_P,k_V)=(\cav g_{k_P,k_V})(p_0)\).
\item There exists an optimal \(\tau^*(k_P,k_V)\) supported on at most two posteriors.
\item Pooling is optimal iff \((\cav g_{k_P,k_V})(p_0)=g_{k_P,k_V}(p_0)\), in which case \(\tau^*=\delta_{p_0}\).
\end{enumerate}
\end{lemma}

Lemma \ref{thm:concav} reduces the outer problem to a one-dimensional concavification exercise. Once we fix the effective posterior payoff $g_{k_P,k_V}$, Bayes plausibility restricts the principal to
posterior distributions with mean $p_0$, and the maximal attainable value at the prior is
exactly the concave envelope $(\cav g_{k_P,k_V})(p_0)$. In the binary-state case, any point
on this envelope can be implemented by a mixture of at most two posteriors $(p_-^\ast,p_+^\ast)$ with $p_-^\ast \le p_+^\ast$ , so there is
always an optimal binary experiment. Pooling is optimal precisely when $g_{k_P,k_V}$ is
already locally concave at the prior, so dispersing beliefs yields no gain and
$\tau^\ast=\delta_{p_0}$. A key implication for our two-layer setting is that comparative
statics operate through the \emph{shape} of $g_{k_P,k_V}$: the inner management problem can
change curvature and introduce or remove kinks in $\phi_{k_V}$, while $k_P$ scales the
convex cost $\kappa$. Because the optimal signal is pinned down by the supporting chord(s)
of $\cav g_{k_P,k_V}$ at $p_0$, small parameter changes can reselect the contact points and
generate discrete jumps in $(p_-^\ast,p_+^\ast)$ and hence in informativeness.

\subsection{Monotone informativeness in the two--posterior domain}

We now consider comparative statics implications for information acquisition in the two--posterior domain. In particular, we want to understand when information and management are \emph{complements} or \emph{substitutes} as the cost of each instrument varies.

Let \(\mathcal T_2\) be the set of Bayes--plausible distributions \(\tau\) over posteriors
\(p\in[0,1]\) with mean \(p_0\) and support size at most two. Any \(\tau\in\mathcal T_2\) can be written as
\(
\tau=\alpha\,\delta_{p_-}+(1-\alpha)\,\delta_{p_+},
\;
0\le p_-\le p_0\le p_+\le 1,
\;
\alpha=\frac{p_+-p_0}{p_+-p_-},
\)
with the degenerate case \(\tau=\delta_{p_0}\) corresponding to \(p_-=p_+=p_0\).
We endow \(\mathcal T_2\) with the weak topology. If $\tau=\alpha\delta_{p_-}+(1-\alpha)\delta_{p_+}$ and $\tau'=\alpha'\delta_{p_-'}+(1-\alpha')\delta_{p_+'}$,
then their join $\tau\join\tau'$ has endpoints
\(
\big(\min\{p_-,p_-'\},\ \max\{p_+,p_+'\}\big),
\)
and their meet $\tau\meet\tau'$ has endpoints
\(
\big(\max\{p_-,p_-'\},\ \min\{p_+,p_+'\}\big).
\)

\emph{Convex order on \(\mathcal T_2\).}
For \(\tau,\tau'\in\mathcal T_2\), write \(\tau'\succeq_{cx}\tau\) if
\(\int\psi\,d\tau'\ge \int\psi\,d\tau\) for every convex \(\psi:[0,1]\to\mathbb R\).
On \(\mathcal T_2\), convex order reduces to endpoint nesting: if
\(p_-,p_+\in \supp(\tau)\) and \(p_-',p_+'\in\supp(\tau')\), then $\tau'\succeq_{cx}\tau
\;\text{iff}\;\;
p_-'\le p_- \ \text{and}\ p_+'\ge p_+ .$ Hence \((\mathcal T_2,\succeq_{cx})\) is a lattice.

\emph{Strong set order on \(\mathcal T_2\).}
For nonempty sets $A,B\subseteq \mathcal T_2$, define
\(A \SSO B\) if for all \(a\in A\) and for all \(b\in B\), \(a\join b\in A\) and \(a\meet b\in B\). Thus $A\SSO B$ means $A$ lies (weakly) \emph{above} $B$ in the lattice in Topkis' strong set order (\cite{Topkis1998}).

\subsubsection{Comparative statics in the information-cost parameter}

Fix $k_V\ge 0$ and consider the outer problem restricted to $\mathcal T_2$:
\[
\max_{\tau\in\mathcal T_2}\;U_{k_V}(\tau \mid k_P),
\; \text{where} \;\;
U_{k_V}(\tau \mid k_P):=\int_0^1 \phi_{k_V}(p)\,\tau(dp)\;-\;k_P\int_0^1 \kappa(p)\,\tau(dp).
\]

The key observation is that $k_P$ scales a \emph{convex} functional of the posterior.
Hence, as $k_P$ rises, more dispersed signals become relatively more expensive, suggesting
that $U_{k_V}$ has \emph{decreasing differences} in $(\tau,k_P)$ with respect to
$\succeq_{cx}$.

\begin{proposition}\label{lem:DD_kP_motiv}
For any $k_P'>k_P$ and any $\tau',\tau\in\mathcal T_2$ with $\tau'\succeq_{cx}\tau$,
\[
U_{k_V}(\tau' \mid k_P')-U_{k_V}(\tau \mid k_P')
\ \le\
U_{k_V}(\tau' \mid k_P)-U_{k_V}(\tau \mid k_P).
\]
\end{proposition}

Proposition~\ref{lem:DD_kP_motiv} formalizes a simple force: if $\tau'$ is a mean-preserving spread
of $\tau$, then $\E_{\tau'}[\kappa(p)]\ge \E_{\tau}[\kappa(p)]$ because $\kappa$ is convex. Thus raising $k_P$ disproportionately penalizes more informative experiments, making them
less attractive at the margin. Decreasing differences implies monotone comparative statics for ordered choices such that optimal informativeness is decreasing in $k_P$, suggesting a monotone \emph{inward} response of optimal posteriors. To show this result clearly, we focus on a totally ordered subset of \(\mathcal T_{2}\) which we now define.

Fix endpoints $a\in[0,p_0)$ and $b\in(p_0,1]$. Let $\bar\lambda \in (0,1]$ and for each $\lambda\in[0,\bar\lambda]$, define the two posterior points
\(p_-(\lambda):=p_0-\lambda(p_0-a),\;p_+(\lambda):=p_0+\lambda(b-p_0).\) Let \( \tau(\lambda) :=\alpha\,\delta_{p_-(\lambda)}+(1-\alpha)\,\delta_{p_+(\lambda)}\) be the corresponding information-policy where \(\alpha:=\frac{b-p_0}{b-a}\) is fixed for each policy. Let \( \mathcal C:=\{\tau(\lambda):\lambda\in[0,\bar\lambda]\}\subset \mathcal T_2 \) denote the resulting $\lambda$-chain. For $\lambda',\lambda\in[0,\bar\lambda]$, we have
\( \lambda'\ge \lambda \;\text{iff}\; \tau(\lambda')\ \cx\ \tau(\lambda).\) As such, informativeness and dispersion coincide for information-policies on $\mathcal C$. In particular, $(\mathcal C,\cx)$ is totally ordered (a chain). For any $\lambda,\lambda'\in[0,\bar\lambda]$,
\( \tau(\lambda)\vee\tau(\lambda')=\tau(\max\{\lambda,\lambda'\}),
\;
\tau(\lambda)\wedge\tau(\lambda')=\tau(\min\{\lambda,\lambda'\}).
\)
Thus $\mathcal C$ is closed under $\vee,\wedge$ (a sublattice of $(\mathcal T_2,\cx)$). Let $M^{\mathcal C}_{k_V}(k_P):=\arg\max_{\tau\in\mathcal C} U_{k_V}(\tau \mid k_P)$ denote the set of maximizers of the outer problem when the domain is restricted to $\mathcal C$. We have the following comparative statics result.

\begin{proposition}\label{prop:MCS_set}
For any $k_P'>k_P$, $M^{\mathcal C}_{k_V}(k_P) \SSO M^{\mathcal C}_{k_V}(k_P')$. In particular, for any pair of \(\tau \in M^{\mathcal C}_{k_V}(k_P)\) and \(\tau' \in M^{\mathcal C}_{k_V}(k_P')\) there exist $\hat\tau\in M^{\mathcal C}_{k_V}(k_P)$ and $\hat\tau'\in M^{\mathcal C}_{k_V}(k_P')$ such that
\(
\hat\tau \ \succeq_{cx}\ \hat\tau'.
\)
\end{proposition}

Proposition~\ref{prop:MCS_set} says that costlier information pushes the optimal
experiment toward \emph{pooling}: the low posterior rises, the high posterior falls, and the
signal becomes less dispersed. Similarly, cheaper information shifts the optimal experiment towards more dispersion: the low posterior falls, the high posterior rises, and the
signal becomes more dispersed. As such, Proposition~\ref{prop:MCS_set} provides a direct link between the geometric movement of
contact points and costs of informativeness: higher $k_P$ reduces posterior
spread and lower  $k_P$ increases posterior
spread. Economically, when generating dispersion is expensive, the
principal relies more on the inner instrument (management) at beliefs close to $p_0$ and less
on pushing probability mass to extreme posteriors. By contrast, when generating the dispersion is cheaper, the
principal relies more on the outer instrument (information) by pushing the probability mass to the extreme posteriors. As a result, when the cost of information changes, management and information behave as ``substitutes''.

\subsubsection{Comparative statics in the management-cost parameter}

Fix $k_P$ and define
\(
U_{k_P}(\tau\mid k_V)
:=\int_0^1 \phi_{k_V}(p)\,\tau(dp)-k_P\int_0^1 \kappa(p)\,\tau(dp),
\;\tau\in\mathcal T_2
\). Let $M^{\mathcal C}_{k_P}(k_V):=\arg\max_{\tau\in \mathcal C} U_{k_P}(\tau \mid k_V)$ denote the set of maximizers of the outer problem restricted to the set $\mathcal C$. 

\emph{Complements vs substitutes.}
We say that information and management are \emph{set-valued complements} on $\mathcal C$ if
\(
k_V'<k_V\ \Longrightarrow\ M^{\mathcal C}_{k_P}(k_V') \ \SSO\ M^{\mathcal C}_{k_P}(k_V),
\)
and \emph{set-valued substitutes} if the reverse order holds:
\(
k_V'<k_V\ \Longrightarrow\ M^{\mathcal C}_{k_P}(k_V) \ \SSO\ M^{\mathcal C}_{k_P}(k_V').
\)

We impose an approximation requirement:

\begin{enumerate}[label=(R),leftmargin=18pt]
\item \emph{Richness for maximizers.} For any $\tau',\tau\in\mathcal C$ with
$\tau'\succeq_{cx}\tau$ and any $k_V'<k_V$, there exist sequences $k_{V,n}\to k_V$ and
$k_{V,n}'\to k_V'$ with $k_{V,n}'<k_{V,n}$ and selections
\(
\sigma_n\in M^{\mathcal C}_{k_P}(k_{V,n}),\; \sigma_n'\in M^{\mathcal C}_{k_P}(k_{V,n}'),
\)
such that $\sigma_n\to\tau$ and $\sigma_n'\to\tau'$ in $\mathcal C$.
\end{enumerate}

For $k_V'<k_V,\;$let
$
\Delta\phi_{k_V',k_V}(p):=\phi_{k_V'}(p)-\phi_{k_V}(p)
$ denote the \emph{gain from cheaper management}.
Since $k_P$ is fixed, the information-cost term cancels in differences across $k_V$:
\[
U_{k_P}(\tau\mid k_V')-U_{k_P}(\tau\mid k_V)
=\int_0^1 \Delta\phi_{k_V',k_V}(p)\,\tau(dp).
\]

Assume that $c_V$ is continuous and so $\Delta\phi_{k_V',k_V}(p)$ is continuous in $(k_V',k_V)$. The following result characterizes when information and management are complements or substitutes.

\begin{proposition}\label{prop:comp_subs_necessity}
The following are equivalent:
\begin{enumerate}[label=\textup{(\roman*)},leftmargin=22pt]
\item \emph{Complementarity (resp., substitutability):} for all $k_V'<k_V$,
\[
M^{\mathcal C}_{k_P}(k_V') \ \SSO\ M^{\mathcal C}_{k_P}(k_V)
\qquad
(\text{resp., } M^{\mathcal C}_{k_P}(k_V)\ \SSO\ M^{\mathcal C}_{k_P}(k_V')).
\]
\item \emph{Increasing (resp., decreasing) differences:} for all $k_V'<k_V$ and all
$\tau'\succeq_{cx}\tau$ in $\mathcal C$,
\[
\int \Delta\phi_{k_V',k_V}\,d\tau'
\ \ge\
\int \Delta\phi_{k_V',k_V}\,d\tau \qquad \Big(\text{resp.,}
\int \Delta\phi_{k_V',k_V}\,d\tau'
\ \le\
\int \Delta\phi_{k_V',k_V}\,d\tau. \Big)
\]
\end{enumerate}

\end{proposition}

Proposition~\ref{prop:comp_subs_necessity} connects a set-valued
notion of complementarity/substitutability to a checkable
condition on the marginal value of cheaper management across posteriors.  On $\mathcal C$,
informativeness is totally ordered by $\cx$, so complementarity means that lowering $k_V$ shifts the
entire optimal set of experiments outward in the strong set order (more dispersed posteriors
become optimal), whereas substitutability means the opposite. Part (ii) pinpoints the mechanism:
if $\Delta\phi_{k_V',k_V}$ is larger under more dispersed posterior distributions (i.e., its
expectation is increasing in $\cx$), then experiments that generate more extreme posteriors benefit
more from cheaper management, so the principal optimally pairs cheaper management with more
informative signals (complements). If instead $\Delta\phi_{k_V',k_V}$ is relatively concentrated on
moderate posteriors, then cheaper management mainly raises the payoff of less informative
experiments, producing substitutes. Practically, the proposition is useful because it lets
one diagnose complementarity/substitutability by studying the monotonicity of
$\int \Delta\phi_{k_V',k_V}\,d\tau$ on $\mathcal C$ rather than solving the full outer
problem for every $k_V$, and it remains valid even when the outer optimum is not unique since it
works directly with the maximizer correspondence and the strong set order.

\emph{Discussion.} In many applications the principal cannot freely choose an arbitrary signal structure, but can
instead tune the \emph{intensity} (precision, resolution, or thoroughness) of a fixed evidence technology.
The set $\mathcal C=\{\tau(\lambda):\lambda\in[0,\bar\lambda]\}$ formalizes this idea in the binary-state,
two-posterior domain: the endpoints $a<p_0<b$ represent the most extreme posteriors attainable from
the two canonical messages (``bad news'' and ``good news''), while the scalar $\lambda$ scales how
decisive each message is. At $\lambda=0$ the policy is uninformative (pooling at $p_0$); as $\lambda$
increases the two posteriors move monotonically away from $p_0$ along fixed directions, generating
greater dispersion and informativeness. The weight $\alpha$ is held fixed so Bayes plausibility is
satisfied for every $\lambda$ and the frequency of the two messages is not itself a design choice along
the chain. As such, it is
empirically meaningful to compare policies by a single informativeness index; on $\mathcal C$ this index
coincides with the convex order, so dispersion and informativeness are aligned.

\subsubsection{Direct vs. contact-set channels}\label{subsubsec:431}

In this subsection, we discuss how comparative statics of the optimal information policy $\tau^*(k_P,k_V)$ can be decomposed into two logically distinct mechanisms.
The key observation is that $\tau^*$ is not chosen to maximize $g_{k_P,k_V}(p)$ pointwise in $p$, but rather to maximize its \emph{Bayes--plausible expectation}
via concavification (Lemma~\ref{thm:concav}).
As a result, changes in $(k_P,k_V)$ affect $\tau^*$ both through how payoffs change at each posterior and through how the \emph{supporting geometry} changes.

\emph{The direct (pointwise) channel.}
Fix any feasible information policy $\tau$ with mean $p_0$.
A change in $k_V$ affects the principal's payoff under this fixed $\tau$ only through the function $g_{k_P,k_V}$:
\(
\E_{\tau}\!\big[g_{k_P,k_V}(p)\big]
=
\int_0^1 \Big(\phi_{k_V}(p)-k_P\kappa(p)\Big)\,\tau(dp).
\)
By Lemma~\ref{thm:bangbang}, $\phi_{k_V}(p)$ is weakly decreasing in $k_V$ for each $p$, and management at $p\in(\pi_L,\pi_H)$ switches off once $k_V$ exceeds the posterior-specific break-even level.
Thus, \emph{for any fixed $\tau$}, increasing $k_V$ weakly reduces the contribution of those posteriors where management is used, typically compressing payoff differences across posteriors.
Similarly, increasing $k_P$ lowers $g_{k_P,k_V}(p)$ more in the tails because $\kappa$ is convex.
The direct channel therefore captures the most immediate intuition:
\emph{holding the information structure fixed}, higher $k_V$ makes management less valuable and higher $k_P$ makes dispersion more expensive, both of which tend to reduce the gain from informative signals.

\emph{The contact-set (supporting-line) channel.}
While the direct channel holds $\tau$ fixed, the optimal policy $\tau^{\ast}(k_P,k_V)$ is
chosen by solving
$
\max_{\tau:\,\E_\tau[p]=p_0}\ \E_\tau\!\big[g_{k_P,k_V}(p)\big].
$
In the binary-state case, Lemma~\ref{thm:concav} implies
$
U(k_P,k_V)=(\cav g_{k_P,k_V})(p_0)
\;\text{and}\;
\supp(\tau^{\ast})\subseteq\{p_-(k_P,k_V),\,p_+(k_P,k_V)\},
$
where $\{p_-,p_+\}$ are the contact points at which a supporting chord of $\cav g_{k_P,k_V}$
through $p_0$ touches $g_{k_P,k_V}$. Equivalently, $(p_-,p_+)$ are the posteriors that
implement the optimal concavifying mixture at the prior. Crucially, when parameters change it is not only the level of $g_{k_P,k_V}$ that moves; its
\emph{shape} can change. Changes in $k_V$ affect the inner value $\phi_{k_V}$: increasing
$k_V$ shrinks the region where management is used (Corollary~\ref{cor:phat}) and can create or remove
kinks in $\phi_{k_V}(\cdot)$ at posteriors where the optimal management switches between $0$
and $q_{\min}(p)$. Changes in $k_P$ reshape $g_{k_P,k_V}$ in a different way: since
$g_{k_P,k_V}(p)=\phi_{k_V}(p)-k_P\kappa(p)$ and $\kappa$ is convex, increasing $k_P$ scales a
convex penalty and thereby alters the global curvature of the effective payoff (often making
it more concave and compressing relative payoffs away from the prior). Because concavification
is determined by global curvature and kinks, small changes in either $k_V$ or $k_P$ can
trigger \emph{contact-set reselection}, causing the supporting chord at $p_0$ to switch to a
different pair of tangency points. Hence $(p_-(k_P,k_V),p_+(k_P,k_V))$---and thus $\tau^{\ast}$---can move
discontinuously even when primitives vary smoothly, with possible jumps between interior
contact pairs, boundary contact points, and pooling.

\emph{Discussion.}
The direct channel is typically monotone: higher $k_V$ lowers $\phi_{k_V}(p)$ pointwise and higher $k_P$ penalizes dispersion pointwise through convex $\kappa$.
If $\tau$ were exogenously fixed, informativeness could only become less attractive.
However, the contact-set channel can generate retargeting and jumps:
as $k_V$ increases, the principal may abandon moderately informative posteriors and instead choose rarer but more decisive posteriors (more dispersion),
or may compress posteriors toward the prior (less dispersion), depending on how the supporting chord re-selects its tangency points, as we characterized in Proposition \ref{prop:comp_subs_necessity}. Appendix \ref{app:worked_examples} provides numerical examples which illustrate non-monotone changes in informativeness as $k_V$ changes. On the other hand, higher $k_P$ reduces posterior spread since the principal relies more on the inner instrument (management) at beliefs close to $p_0$ and less on pushing probability mass to extreme posteriors (information), as we showed in Proposition \ref{prop:MCS_set}. 

\subsection{Interdependence between information and management}
The framework emphasizes a two-way interaction between information acquisition and bias management. In this subsection, we summarize the main regimes in information and management, derive a simple two-threshold characterization, and discuss regulatory implications through the perspective of interdependence.

\subsubsection{Qualitative regimes and a limiting benchmark}
A useful empirical prediction is that management matters only when posteriors fall in the disagreement region, and information is valuable mainly insofar as it generates such posteriors and separates beliefs into regions where different actions are optimal. Outside that region---when posteriors are already extreme or when management is either very cheap or relatively infeasible---one of the instruments becomes redundant and the other dominates the design problem. The possible regimes in the conflict region $(\pi_L,\pi_H)$ can be summarized as follows. 

\emph{Low $k_V$.}
Management can cheaply correct the delegated decision at posteriors near $p_0$, shrinking the concavification gain from splitting beliefs.
Pooling $\tau^*=\delta_{p_0}$ can therefore be optimal even when informative signals would otherwise be valuable.

\emph{Intermediate $k_V$.}
As $k_V$ rises, management becomes too expensive at some posteriors.
The principal may then use information to create posteriors at which either (i) little management is needed or (ii) the action changes without management.
In this region, $\tau^*$ is informative and the realized posteriors $(p_-,p_+)$ can \emph{retarget} (sometimes discontinuously) as $k_V$ varies,
because the supporting chord defining $\cav g$ at $p_0$ can switch contact points.

Since management is chosen after the posterior is realized, the relevant choices are \(q^*(p_-;k_V)\) and \(q^*(p_+;k_V)\).
Corollary~\ref{cor:phat} implies $q^*(p_\pm;k_V)>0 \;\text{iff}\; p_\pm\ge \hat p(k_V),$ when \(p_\pm\in(\pi_L,\pi_H)\); if \(p\notin(\pi_L,\pi_H)\), management is irrelevant. Thus, the use of management in the conflict region after an informative signal can be classified as: (i) management at both posteriors (\(p_-\ge \hat p(k_V)\)); (ii) management only at the high posterior (\(p_-<\hat p(k_V)\le p_+\)); (iii) no management (\(p_+<\hat p(k_V)\)).

\emph{High $k_V$.}
Beyond a (possibly infinite) level $k_V^{NM}$, management is never optimal at any posterior in $(\pi_L,\pi_H)$,
so the problem reduces to a no-management persuasion problem with cost $k_P\,\kappa(.)$.
Depending on $k_P$ and primitives, $\tau^*$ may be informative or pooling in this limiting regime. 

\subsubsection{Two thresholds in the management cost}\label{subsubsec:regimes_two_thresholds}

We now present a simple two-threshold result that supports the emergence of the regimes described above: one threshold at which information first becomes optimal and one threshold beyond which management is never optimal at any posterior. Fix $k_P>0$.

We assume the following four conditions: (A1)  $C$ is bounded from above and $k_V\in[0,k]$ for some $k > 0$; (A2) Pooling is optimal at $k_V=0$, i.e.\
$(\cav g_{k_P,0})(p_0)=g_{k_P,0}(p_0)$; (A3) There exists some $\hat k_V\in[0,\bar k]$ such that informativeness is optimal at $\hat k_V$, i.e.\
$(\cav g_{k_P,\hat {k}_V})(p_0)>g_{k_P,\hat {k}_V}(p_0)$; (A4) For all $p\in[0,1]$, the break-even cost $\bar k_V(p)$ is bounded by $k>0$.

Let $k_V^{\mathrm{ON}}(k_P)\ :=\ \inf\Big\{k_V \in [0,\bar k]:\ (\cav g_{k_P,k_V})(p_0)>g_{k_P,k_V}(p_0)\Big\}$ be the \emph{information-on threshold} and $k_V^{NM}:=\sup_{p\in(\pi_L,\pi_H)} \bar k_V(p) \in[0,k]$ denote the global \emph{no-management threshold}. By (A4), $k_V^{NM}\leq k$. \bigskip  

Suppose (A1)--(A4) hold. \medskip

\begin{proposition}\label{thm:two_thresholds}
The following are true:
\begin{enumerate}[label=\textup{(\roman*)},leftmargin=22pt]
\item If \(k_V<k_V^{\mathrm{ON}}(k_P)\), pooling is optimal.
\item For every \(\varepsilon>0\), some \(k_V\in\big(k_V^{\mathrm{ON}}(k_P),\,k_V^{\mathrm{ON}}(k_P)+\varepsilon\big)\) admits an informative optimum.
\item If \(k_V\ge k_V^{NM}\), then \(q^*(p;k_V)=0\) for all \(p\in[0,1]\).
\item If \(k_V<k_V^{NM}\), then \(q^*(p;k_V)>0\) for some \(p\in(\pi_L,\pi_H)\).
\end{enumerate}
\end{proposition}

The threshold $k_V^{\mathrm{ON}}(k_P)$ is a \emph{first-on} threshold: it guarantees that pooling is optimal below it and that
informativeness is optimal above it for some $k_V$, under (A1)--(A4).
It does \emph{not} assert that informativeness is monotone in $k_V$ thereafter; further increases in $k_V$ can still induce re-targeting of contact points
and, depending on $k_P$, can eventually restore pooling. In contrast, $k_V^{NM}$ is a global threshold: it is the maximal ``break-even'' management cost across posteriors in the disagreement region $(\pi_L,\pi_H)$.
When $k_V$ exceeds this level, even the most favorable posterior for management (highest gain relative to required intervention cost) cannot justify intervention, so the principal optimally abandons bias management entirely. Appendix \ref{app:worked_examples} examples illustrate the thresholds stated in Proposition \ref{thm:two_thresholds} with a finite and an infinite global threshold $k_V^{NM}$.

\subsubsection{Application: subsidizing information vs subsidizing management}\label{subsec:regulator}

Suppose a regulator can influence the
principal's ``steering technology'' by reducing either the cost of information acquisition
($k_P$) or the cost of bias management ($k_V$), but can subsidize only one parameter due to
budgetary or institutional constraints. The regulator's objective is to facilitate more
steering---either by increasing informativeness of the principal's signal, by increasing
bias-management effort, or by balancing the two---net of the cost of intervention. Because
information and management interact through the posterior distribution, the choice of which
parameter to subsidize depends both on the regulator's goal and on whether the principal's
instruments behave as \emph{complements} or \emph{substitutes} in equilibrium. What should be the optimal policy?

Two forces guide the policy ranking. First, holding $k_V$ fixed, lowering $k_P$ directly
reduces the convex penalty $k_P\kappa(\cdot)$ and therefore makes the optimal experiment weakly more informative: the optimal contact points
move outward and pooling becomes less likely. Second, holding $k_P$ fixed, lowering $k_V$
raises the posterior-by-posterior return to management and can either increase or decrease
informativeness depending on whether the model exhibits complementarity or substitutability.
In complement regions, cheaper management raises the marginal value of dispersion; in
substitute regions, cheaper management ``fixes'' choices locally and dampens the marginal
value of belief dispersion. We can classify optimal policies by regulatory objective as follows:

\emph{(i) More information.}
If the primary objective is to increase informativeness,
the robust instrument is to subsidize information acquisition (reduce $k_P$). This policy
acts directly on the outer problem by lowering the marginal cost of dispersion and, in the
binary-state case, expands the concavification gap at the prior. Subsidizing management
(reducing $k_V$) can also raise informativeness, but only in environments where information
and management are complements; in substitute environments it may reduce informativeness by
making local management a cheaper substitute for belief dispersion.

\smallskip
\emph{(ii) More management.}
If the objective is to increase bias-management effort (e.g., governance, enforcement
capacity, incentive alignment), the direct instrument is to subsidize management (reduce
$k_V$). This raises $q^*(p)$ at posteriors in the conflict region and weakly enlarges the
set of posteriors at which management is optimally used. By contrast, subsidizing information
(reducing $k_P$) may reduce the \emph{incidence} of management by shifting posterior mass
toward regions where tastes agree and management is irrelevant.

\smallskip
\emph{(iii) More information and more management.}
Here the complement/substitute classification is pivotal. If the environment is in a
\emph{complement} regime, then reducing $k_V$ can ``crowd in'' both instruments: management
intensity rises and the optimal signal becomes more informative. In contrast, in a
\emph{substitute} regime, there is an intrinsic tradeoff: reducing $k_P$ tends to raise
informativeness but can reduce management usage, whereas reducing $k_V$ tends to raise
management but can reduce informativeness. With only one policy lever, the regulator must
therefore choose which margin is more valuable (or which can be adjusted more cheaply), since
a single subsidy cannot generally raise both dimensions simultaneously when the instruments
substitute.

These recommendations are strongest when the prior lies in the conflict region
$p_0\in(\pi_L,\pi_H)$, where both instruments can matter. If $p_0\notin(\pi_L,\pi_H)$, the
agent's preferred action is typically invariant to taste at beliefs near the prior, so
bias management has little bite; in that case subsidizing $k_P$ is the natural way to
increase steering via information. More generally, the regulator's optimal intervention can
be viewed as targeting the scarcest input to steering: subsidize $k_P$ when dispersion is
valuable but too costly to generate, and subsidize $k_V$ when alignment is pivotal but too
costly to enforce.

\section{Conclusion}\label{sec:conclusion}

In this paper, we study a principal who delegates choice to a biased agent and can steer outcomes using two costly instruments:
an information policy that determines a Bayes--plausible distribution of posteriors, and a bias-management action that shifts the agent's effective taste within the convex hull of two benchmark tastes. In a binary-state, two-action specialization with posterior-separable information costs, the problem admits a clean inner--outer decomposition. Optimal bias management is bang--bang: at each posterior, the principal either does nothing or applies the minimal intervention needed to induce the desired action. The information problem reduces to concavification of an endogenous posterior value function that already incorporates the management option and net information costs, implying that optimal signals use at most two posteriors.

The framework clarifies how information acquisition and bias management are interrelated:
information changes which posteriors occur and thus where management is exercised, while management reshapes the curvature and kinks of the posterior value function and thus the marginal value of information. This feedback implies that the natural benchmark of \emph{monotone informativeness} need not hold. As $k_V$ changes, the inner management problem can switch regimes, reshaping $g_{k_V,k_P}(p)=\phi_{k_V}(p)-k_P\kappa(p)$; the resulting concavification may then reselect its contact set at the prior, generating discrete jumps in optimal posteriors and episodes of inward or outward ``retargeting.'' When regime switching in the inner problem creates kinks or
nonconvexities in $\Delta\phi$, the model can generate mixtures of complementarity and
substitution across parameter ranges, consistent with the nonmonotone patterns observed in
our numerical examples given in Appendix \ref{app:worked_examples} below. As such, when considering policy-implications of the two-layer model, understanding the nature of interaction between information and management--whether they are complements or substitutes--is crucial to target the right lever for an efficient policy making. Overall, the two-layer delegation framework is both empirically natural and theoretically rich, yielding comparative-statics patterns that are relevant for policy analysis of belief-based and rule-based steering tools.

\appendix

\section{Appendix}
\subsection{Numerical examples}\label{app:worked_examples}
In this section,  we collect two numerical examples in the binary-state, two-action, convex-hull taste specialization.

In both examples, states are $S=\{H,L\}$ with prior $p_0=0.5$, and cutoff posteriors $\pi_L=0.3$ and $\pi_H=0.7$. The principal chooses an information policy $\tau\in\Delta_0([0,1])$ and a posterior-contingent
bias-management intensity $q\in[0,1]$.
The binary-state outer problem is solved by concavification of $g_{k_P,k_V}(p)$ and an optimal $\tau^*(k_V)$ can be taken to have support size at most two:
\(
\tau^*(k_V)=\alpha(k_V)\,\delta_{p_-(k_V)}+\big(1-\alpha(k_V)\big)\,\delta_{p_+(k_V)},
\;
\alpha(k_V)=\frac{p_+(k_V)-p_0}{p_+(k_V)-p_-(k_V)}.
\)
We represent the optimal posterior spread by
\(
\mathrm{Disp}(k_V):=p_+(k_V)-p_-(k_V),\) with \(\mathrm{Disp}(k_V)=0\text{ under pooling }(p_-=p_+=p_0).\)
We measure informativeness in terms of \emph{mutual information} $I(\tau(k_V))$ in \emph{bits} for the induced two-signal experiment
supporting posteriors $\{p_-,p_+\}$ with weights $\{\alpha,1-\alpha\}$:
\(I(\tau(k_V))=H(p_0)-\big(\alpha H(p_-)+(1-\alpha)H(p_+)\big),\)
where
$H(p)=-p\log_2 p-(1-p)\log_2(1-p)$ denotes the entropy function.

Outcomes are $X=\{x,y\}$, and the principal's utility is normalized as $u(x)=1$ and $u(y)=0$.
The menu is $A=\{f,g\}$, interpreted as \emph{status quo} ($f$) and \emph{intervention} ($g$).
Acts are specified by success probabilities $\Pr(x\mid a,s)$, so $u(a)\cdot p=\Pr(x\mid a,p)$.
Each taste $v_j$ evaluates $g$ as the principal payoff net of a private cost $c_j\ge 0$: $v_j(f)\cdot p=u(f)\cdot p$ and $v_j(g)\cdot p=u(g)\cdot p-c_j$.
Thus $v_j$ chooses $g$ iff $\Delta_u(p):=u(g)\cdot p-u(f)\cdot p\ge c_j$, i.e.\ via a cutoff in $p$.
Thus for the given $\pi_L<\pi_H$, set $c_L=\Delta_u(\pi_L)$ and $c_H=\Delta_u(\pi_H)$ once the acts are specified, so that
$v_L$ chooses $g$ iff $p\ge\pi_L$ and $v_H$ chooses $g$ iff $p\ge\pi_H$.
The principal can implement an effective taste in the convex hull: $v_q=q\,v_L+(1-q)\,v_H$ for $q\in[0,1]$ at a cost $c_V(q)$. For $p\in[\pi_L,\pi_H]$, the minimal intensity that induces $g$ is
\(
q_{\min}(p)\ :=\ \frac{c_H-\Delta_u(p)}{c_H-c_L}\in[0,1],
\quad\text{so }q_{\min}(\pi_L)=1,\ q_{\min}(\pi_H)=0,
\)
extended by clipping to $[0,1]$ outside $[\pi_L,\pi_H]$. \medskip

\begin{example}[Non-monotone informativeness; single peak]\label{ex:pool_info_pool}
Let acts satisfy 
$\Pr(x\mid f,H)=0.4,\; \Pr(x\mid f,L)=0.2,\;\Pr(x\mid g,H)=0.9,\; \Pr(x\mid g,L)=0.5,$
so
$u(f)\cdot p=0.2+0.2p,\quad u(g)\cdot p=0.5+0.4p,\quad \Delta_u(p)=0.3+0.2p.$
Then $c_L=\Delta_u(\pi_L)=0.36$ and $c_H=\Delta_u(\pi_H)=0.44$, generating the cutoffs $\pi_L<\pi_H$. When used, management incurs a fixed-plus-quadratic cost according to the lower semi-continuous cost function $c_V(q)$ satisfying $c_V(0)=0$ and $c_V(q)=\varepsilon+q^2$ for $q>0$ with $\varepsilon=0.03$. 
At $p\in(\pi_L,\pi_H)$ the principal either sets $q^*(p)=0$ or $q^*(p)=q_{\min}(p)$ depending on whether
$\Delta_u(p)\ge k_V(0.03+q_{\min}(p)^2)$.
For $p\le\pi_L$ and $p\ge \pi_H$, choices are unanimous and $q^*(p)=0$. Let the information cost be posterior-separable with
$\kappa(p)=(p-p_0)^2,\; k_P=11.$

All primitives above are held fixed; only $k_V$ varies. Using numerical concavification on a fine grid of $p\in[0,1]$, the optimal contact points, weight, dispersion, informativeness, and management intensities are:

\begin{center}
\begin{tabular}{c|ccccccc}
\hline
$k_V$ & $p_-(k_V)$ & $p_+(k_V)$ & $\alpha(k_V)$ & $\mathrm{Disp}(k_V)$ & $I(\tau(k_V))$ & $q^*(p_-)$ & $q^*(p_+)$\\
\hline
$0.90$ & $0.5000$ & $0.5000$ & $1.0000$ & $0$ & $0$ & $0.5000$ & $0.5000$\\
$0.93$ & $0.3873$ & $0.5014$ & $0.0124$ & $0.1141$ & $0.0005$ & $0$      & $0.4965$\\
$2.00$ & $0.4529$ & $0.5886$ & $0.6529$ & $0.1357$ & $0.0121$ & $0$      & $0.2785$\\
$3.50$ & $0.4905$ & $0.6329$ & $0.9333$ & $0.1424$ & $0.0037$ & $0$      & $0.1677$\\
$4.03$ & $0.4997$ & $0.6419$ & $0.9977$ & $0.1423$ & $0.0001$ & $0$      & $0.1452$\\
$4.05$ & $0.5000$ & $0.5000$ & $1.0000$ & $0$ & $0$ & $0$      & $0$\\
\hline
\end{tabular}
\end{center}
\medskip

In this example there are two threshold-like switches (up to numerical tolerance):
\[
k_V^{\mathrm{ON}}\approx 0.9223
\quad\text{and}\quad
k_V^{\mathrm{OFF}}\approx 4.0304,
\]
such that $\tau^*=\delta_{p_0}$ (pooling) for $k_V<k_V^{\mathrm{ON}}$, $\tau^*$ is informative for $k_V\in(k_V^{\mathrm{ON}},\,k_V^{\mathrm{OFF}})$,
and $\tau^*=\delta_{p_0}$ again for $k_V>k_V^{\mathrm{OFF}}$.
Intuitively, for very small $k_V$ the principal prefers to avoid paying the convex information cost $k_P$ and instead relies on cheap local steering;
for intermediate $k_V$ the concavification gap becomes positive and it is optimal to create some dispersion (often combined with small management at an interior posterior);
for large $k_V$ management is effectively too expensive and the remaining informational gains do not justify $k_P$, returning the solution to pooling. As such, dispersion first starts moving upward, then downward; bias-management turns on and is exercised only at high posteriors and then gradually decreases to no management at all.
\end{example} \medskip

\begin{example}[Non-monotone informativeness; double peak]\label{ex:interior_split_pos_mgmt}
Let acts satisfy
$\Pr(x\mid f,H)=0.30,\quad \Pr(x\mid f,L)=0.20,\quad \Pr(x\mid g,H)=0.60,\quad \Pr(x\mid g,L)=0.20.$
Then 
$u(f)\cdot p = 0.20+0.10p,\quad u(g)\cdot p = 0.20+0.40p, \quad \Delta_u(p) = 0.30p.$
Then $c_L=\Delta_u(\pi_L)=0.9$ and $c_H=\Delta_u(\pi_H)=2.1$, generating the cutoffs $\pi_L<\pi_H$. Take a convex management cost $c_V(q)=k_V q^2$ with $k_V\geq0$. Let the information cost be posterior-separable with $\kappa(p)=(p-p_0)^4$ and $k_P=80.$

All primitives above are held fixed; only $k_V$ varies. Using numerical concavification on a fine grid of $p\in[0,1]$, the optimal contact points, weight, dispersion, informativeness, and management intensities are:
\begin{center}
\begin{tabular}{c|ccccccc}
\hline
$k_V$ & $p_-(k_V)$ & $p_+(k_V)$ & $\alpha(k_V)$ & $\mathrm{Disp}(k_V)$ & $I(\tau(k_V))$ & $q^*(p_-)$ & $q^*(p_+)$\\
\hline
$0.05$ & $0.5000$ & $0.5000$ & $1.0000$ & $0$ & $0$ & $1.0000$ & $1.0000$\\
$0.10$ & $0.4625$ & $0.5375$ & $0.5000$ & $0.0750$ & $0.0041$ & $0.6938$ & $0.8063$\\
$0.12$ & $0.4658$ & $0.5342$ & $0.5000$ & $0.0685$ & $0.0034$ & $0.5822$ & $0.6678$\\
$0.20$ & $0.4245$ & $0.7000$ & $0.7258$ & $0.2756$ & $0.0445$ & $0.3183$ & $0$\\
$0.40$ & $0.4115$ & $0.7000$ & $0.6932$ & $0.2885$ & $0.0522$ & $0.1543$ & $0$\\
$0.80$ & $0.4065$ & $0.7000$ & $0.6814$ & $0.2935$ & $0.0551$ & $0.0762$ & $0$\\
\hline
\end{tabular}
\end{center} \medskip

As the bias-management cost parameter $k_V$ increases, the optimal management intensity at any \emph{fixed} posterior in the conflict region $(\pi_L,\pi_H)$ declines smoothly. In the quadratic specification, the inner problem delivers
\(
q^*(p)=\min\Big\{1,\frac{\Delta_u(p)}{2k_V}\Big\}
\;\text{for }p\in(\pi_L,\pi_H),
\)
so (whenever $\Delta_u(p)>0$) one has $q^*(p)>0$ for every finite $k_V$ and $q^*(p)\downarrow 0$ only as $k_V\to\infty$ (management ``fades out'' at a fixed posterior). The information policy adjusts in the opposite direction: when $k_V$ is very small, the principal can rely on inexpensive management and pooling is optimal, yielding $I(\tau(k_V))=0$. Once $k_V$ rises, the principal begins to acquire information and creates belief dispersion, but informativeness need not increase smoothly because concavification contact points can re-select discontinuously. In our calibration, after information first turns on, mutual information exhibits a second (local) dip around $k_V\simeq 0.12$ as the optimal interior--interior split remains close to the prior and compresses slightly. A sharp upward jump then occurs when the optimal supporting chord switches: the upper contact posterior jumps to the decisive cutoff $p_+(k_V)=\pi_H$ which  is around $k_V\simeq 0.15$, producing much greater dispersion (and hence a higher $I(\tau(k_V))$). At the same time, $q^*(p_+)$ drops to zero not because the quadratic inner problem features a finite shutoff threshold at an interior posterior, but because once $p_+=\pi_H$ the preferred action is chosen even under the default taste, so bias management at the upper posterior becomes irrelevant.
\end{example}

\subsection{Proofs}\label{app:proofs}

\begin{proof}[\emph{\textbf{Proof of Proposition \ref{prop:timing}.}}]
Fix any constant $q\in[0,1]$. For every posterior $p$,
\[
\max_{\tilde q\in[0,1]}\big(B(p,\tilde q)-c_V(\tilde q)\big)\ \ge\ B(p,q)-c_V(q).
\]
Integrating both sides against any Bayes--plausible $\tau$ and subtracting $c_P(\tau)$ gives
\[
\int \max_{\tilde q}\big(B(p,\tilde q)-c_V(\tilde q)\big)\,\tau(dp)-c_P(\tau)
\ \ge\
\int B(p,q)\,\tau(dp)-c_P(\tau)-c_V(q).
\]
Maximizing the left-hand side over $\tau$ yields $U^{\mathrm{bas}}$, while maximizing the
right-hand side first over $\tau$ and then over $q$ yields $U^{\mathrm{rev}}$. Hence
$U^{\mathrm{bas}}\ge U^{\mathrm{rev}}$.
\end{proof}
\begin{proof}[\emph{\textbf{Proof of Lemma \ref{thm:bangbang}.}}]
Fix \(p\).
If \(p\le\pi_L\), then \(p\le\pi(q)\) for all \(q\), so \(\Psi(p,q)=u(f)\cdot p-k_VC(q)\), maximized at \(q=0\).
If \(p\ge\pi_H\), then \(p\ge\pi(q)\) for all \(q\), so \(\Psi(p,q)=u(g)\cdot p-k_VC(q)\), maximized at \(q=0\).
If \(p\in(\pi_L,\pi_H)\), the indicator equals \(0\) for \(q<q_{\min}(p)\) and \(1\) for \(q\ge q_{\min}(p)\).
Among choices with indicator \(0\), \(q=0\) is optimal.
Among choices with indicator \(1\), the principal minimizes \(C(q)\) subject to \(q\ge q_{\min}(p)\);
by monotonicity of \(C\) this is achieved at \(q=q_{\min}(p)\).
Comparing the two candidate values yields the stated threshold condition.
\end{proof}
\begin{proof}[\textbf{\emph{Proof of Proposition \ref{lem:kbarMon}.}}]
By assumption, \(\Delta_u(p)\) is increasing.
Since \(q_{\min}(p)\) is decreasing and by assumption \(C\) is increasing,
the denominator \(C(q_{\min}(p))\) is decreasing.
Thus \(\Delta_u(p)/C(q_{\min}(p))\) is increasing where defined; the \(+\infty\) convention handles division-by-zero cases.
\end{proof}
\begin{proof}[\textbf{\emph{Proof of Corollary \ref{cor:phat}.}}]
By Lemma~\ref{thm:bangbang}, for \(p\in(\pi_L,\pi_H)\) management is used iff \(k_V<\bar k_V(p)\).
By Proposition~\ref{lem:kbarMon}, the set \(\{p:\bar k_V(p)\ge k_V\}\) is an interval of the form \([\hat p(k_V),\pi_H]\),
implying the cutoff characterization and monotonicity.
\end{proof}
\begin{proof}[\emph{\textbf{Proof of Lemma \ref{thm:concav}.}}]
The proof follows from standard persuasion/concavification arguments.
Because \(\cav g\) is concave and dominates \(g\), Jensen's inequality implies
\[\int g\,d\tau\le \int \cav g\,d\tau\le (\cav g)(\int p\,d\tau)=(\cav g)(p_0)\] for any feasible \(\tau\).

Achievability follows by taking the supporting line to \(\cav g\) at \(p_0\), selecting its contact points \(p_-\le p_0\le p_+\),
and placing weights on \(\{p_-,p_+\}\) to match the mean \(p_0\).
Pooling is optimal exactly when there is no concavification gap at \(p_0\).
\end{proof}
\begin{proof}[\emph{\textbf{Proof of Proposition \ref{lem:DD_kP_motiv}}}]
Fix $k_P'>k_P$ and $\tau'\succeq_{cx}\tau$. Since $\kappa$ is convex and $\tau',\tau$ share
mean $p_0$, convex-order dominance implies $\int \kappa\,d\tau' \ge \int \kappa\,d\tau$.
Therefore,
\begin{align*}
\big[U_{k_V}(\tau'\mid k_P')-U_{k_V}(\tau \mid k_P')\big]-\big[U_{k_V}(\tau' \mid k_P)-U_{k_V}(\tau \mid k_P)\big]
&=(k_P-k_P')\int \kappa\,(d\tau'-d\tau)\\
&\le 0,
\end{align*} which is the desired inequality.
\end{proof}

\begin{proof}[\textbf{\emph{Proof of Proposition \ref{prop:MCS_set}}}]
Fix $k_V\ge 0$ and write
\(
\tilde U(\lambda,k_P):=U_{k_V}(\tau(\lambda)\mid k_P)\) and \(
\Lambda(k_P):=\arg\max_{\lambda\in[0,\bar\lambda]}\tilde U(\lambda,k_P).
\)
By definition of $M^{\mathcal C}_{k_V}(k_P)$ and of the chain $\mathcal C$,
\begin{equation}\label{eq:Mc_Lambda}
M^{\mathcal C}_{k_V}(k_P)=\{\tau(\lambda):\lambda\in\Lambda(k_P)\}.
\end{equation}
Since $\tau(\lambda)$ is continuous in the weak topology on $\mathcal T_2$ and
$U_{k_V}(\tau\mid k_P)$ is continuous in $(\tau,k_P)$, $\tilde U$ is continuous on the compact set
$[0,\bar\lambda]\times\mathbb R_+$. Hence $\Lambda(k_P)$ is nonempty and compact for each $k_P$.

\emph{Step 1.}
Let $k_P'>k_P$ and $\lambda'\ge \lambda$. Then $\tau(\lambda')\cx\tau(\lambda)$ by construction of the
chain. Applying Proposition~\ref{lem:DD_kP_motiv} with $\tau'=\tau(\lambda')$ and $\tau=\tau(\lambda)$ yields
\begin{equation}\label{eq:DD_lambda_kP}
\tilde U(\lambda',k_P')-\tilde U(\lambda,k_P')
\ \le\
\tilde U(\lambda',k_P)-\tilde U(\lambda,k_P),
\end{equation}
i.e.\ $\tilde U$ has decreasing differences in $(\lambda,k_P)$ on $[0,\bar\lambda]\times\mathbb R_+$.

\emph{Step 2.}
We claim that \eqref{eq:DD_lambda_kP} implies
\begin{equation}\label{eq:SSO_Lambda}
\Lambda(k_P)\ \SSO\ \Lambda(k_P') \qquad\text{for }k_P'>k_P,
\end{equation}
where $\SSO$ is the strong set order on $[0,\bar\lambda]$ induced by the usual order $\ge$, i.e.
$A\SSO B$ iff for all $a\in A$, $b\in B$, $\max\{a,b\}\in A$ and $\min\{a,b\}\in B$.
Since $\Lambda(k_P)$ and $\Lambda(k_P')$ are compact, it is equivalent 
to show
\begin{equation}\label{eq:extremes}
\min\Lambda(k_P)\ \ge\ \min\Lambda(k_P') \qquad\text{and}\qquad
\max\Lambda(k_P)\ \ge\ \max\Lambda(k_P').
\end{equation}

Let $\underline\lambda:=\min\Lambda(k_P)$ and $\overline\lambda':=\max\Lambda(k_P')$.
Suppose for contradiction that $\underline\lambda<\overline\lambda'$. Since
$\underline\lambda\in\Lambda(k_P)$ and $\overline\lambda'\in\Lambda(k_P')$, we have
\begin{equation}\label{eq:opt1}
\tilde U(\underline\lambda,k_P)\ \ge\ \tilde U(\overline\lambda',k_P),
\qquad
\tilde U(\overline\lambda',k_P')\ \ge\ \tilde U(\underline\lambda,k_P').
\end{equation}
Adding the two inequalities in \eqref{eq:opt1} yields
\(
\tilde U(\overline\lambda',k_P')-\tilde U(\underline\lambda,k_P')
\ \ge\
\tilde U(\overline\lambda',k_P)-\tilde U(\underline\lambda,k_P).
\)
But $\overline\lambda'\ge \underline\lambda$ and $k_P'>k_P$, so \eqref{eq:DD_lambda_kP} implies the
reverse inequality:
\(
\tilde U(\overline\lambda',k_P')-\tilde U(\underline\lambda,k_P')
\ \le\
\tilde U(\overline\lambda',k_P)-\tilde U(\underline\lambda,k_P),
\)
hence both inequalities must hold with equality. In particular, \eqref{eq:opt1} must hold with
equality, so $\overline\lambda'$ is also a maximizer at $k_P$ and $\underline\lambda$ is also a
maximizer at $k_P'$, i.e.
\(
\overline\lambda'\in\Lambda(k_P)
\;\text{and}\;
\underline\lambda\in\Lambda(k_P').
\)
This contradicts the definitions $\underline\lambda=\min\Lambda(k_P)$ and
$\overline\lambda'=\max\Lambda(k_P')$ unless $\underline\lambda\ge \overline\lambda'$.
Therefore $\min\Lambda(k_P)\ge \max\Lambda(k_P')$, which in particular implies \eqref{eq:extremes}
(and hence \eqref{eq:SSO_Lambda}). (Equivalently, one can derive separately the two
inequalities in \eqref{eq:extremes} by repeating the same argument with
$\underline\lambda':=\min\Lambda(k_P')$ and $\overline\lambda:=\max\Lambda(k_P)$.) By \eqref{eq:Mc_Lambda} and monotonicity of the map $\lambda\mapsto\tau(\lambda)$ under $\cx$,
\eqref{eq:SSO_Lambda} implies
\(
M^{\mathcal C}_{k_V}(k_P)\ \SSO\ M^{\mathcal C}_{k_V}(k_P').
\)

\emph{Step 3.} Now pick any $\tau\in M^{\mathcal C}_{k_V}(k_P)$ and any $\tau'\in M^{\mathcal C}_{k_V}(k_P')$. We have
\(
\tau\meet\tau' \in M^{\mathcal C}_{k_V}(k_P')
\;\text{and}\;
\tau\join\tau' \in M^{\mathcal C}_{k_V}(k_P).
\)
Let $\hat\tau:=\tau\join\tau'$ and $\hat\tau':=\tau\meet\tau'$.
By the lattice property, $\tau\join\tau'\succeq_{cx}\tau\meet\tau'$, hence
$\hat\tau\succeq_{cx}\hat\tau'$, as desired.
\end{proof}

\begin{proof}[\textbf{\emph{Proof of Proposition \ref{prop:comp_subs_necessity}}}]

Fix $k_P$ and let
\(
M^{\mathcal C}_{k_P}(k_V)=\arg\max_{\tau\in\mathcal C}U(\tau,k_V).
\)
Fix $k_V'<k_V$ and define $\Delta\phi(p):=\Delta\phi_{k_V',k_V}(p)=\phi_{k_V'}(p)-\phi_{k_V}(p)$.
For all $\tau\in\mathcal C$,
\begin{equation}\label{eq:diff_id_pf}
U_{k_P}(\tau,k_V')-U_{k_P}(\tau,k_V)=\int_0^1 \Delta\phi(p)\,\tau(dp).
\end{equation}

We prove the equivalence for the complementarity direction; the substitutability direction is
identical with all inequalities reversed.

\medskip
\noindent\textbf{(ii)$\Rightarrow$(i).}
Assume (ii): for all $\tau'\cx\tau$ in $\mathcal C$,
\begin{equation}\label{eq:ID_pf}
\int \Delta\phi\,d\tau' \ \ge\ \int \Delta\phi\,d\tau.
\end{equation}
We show $M^{\mathcal C}_{k_P}(k_V')\SSO M^{\mathcal C}_{k_P}(k_V)$. Take arbitrary $\sigma'\in M^{\mathcal C}_{k_P}(k_V')$ and $\sigma\in M^{\mathcal C}_{k_P}(k_V)$. Since $\mathcal C$ is a chain, either
$\sigma'\cx\sigma$ or $\sigma\cx\sigma'$. If $\sigma'\cx\sigma$, then on a chain $\sigma'\vee\sigma=\sigma'$ and $\sigma'\wedge\sigma=\sigma$,
so $\sigma'\vee\sigma\in M^{\mathcal C}_{k_P}(k_V')$ and $\sigma'\wedge\sigma\in M^{\mathcal C}_{k_P}(k_V)$ hold immediately. If instead $\sigma\cx\sigma'$, then optimality implies
\begin{equation}\label{eq:opt_pair_pf}
U_{k_P}(\sigma,k_V)\ \ge\ U_{k_P}(\sigma',k_V),
\qquad
U_{k_P}(\sigma',k_V')\ \ge\ U_{k_P}(\sigma,k_V').
\end{equation}
Adding \eqref{eq:opt_pair_pf} and rearranging gives
\(
U_{k_P}(\sigma',k_V')-U_{k_P}(\sigma',k_V)\ \ge\ U_{k_P}(\sigma,k_V')-U_{k_P}(\sigma,k_V).
\)
Using \eqref{eq:diff_id_pf} yields
\(\int \Delta\phi\,d\sigma'\ \ge\ \int \Delta\phi\,d\sigma.\)
But $\sigma\cx\sigma'$ and \eqref{eq:ID_pf} together imply the reverse inequality
$\int \Delta\phi\,d\sigma \ge \int \Delta\phi\,d\sigma'$. Hence equality holds, so
\eqref{eq:opt_pair_pf} must hold with equality. In particular,
$U_{k_P}(\sigma',k_V)=U_{k_P}(\sigma,k_V)$ and $U_{k_P}(\sigma',k_V')=U_{k_P}(\sigma,k_V')$, so $\sigma'\in M^{\mathcal C}_{k_P}(k_V)$ and
$\sigma\in M^{\mathcal C}_{k_P}(k_V')$. Therefore
\(
\sigma'\vee\sigma=\sigma'\in M^{\mathcal C}_{k_P}(k_V')\;\text{and}\;
\sigma'\wedge\sigma=\sigma\in M^{\mathcal C}_{k_P}(k_V),
\)
as required. Since $\sigma'\in M^{\mathcal C}_{k_P}(k_V')$ and $\sigma\in M^{\mathcal C}_{k_P}(k_V)$ were arbitrary, we conclude
$M^{\mathcal C}_{k_P}(k_V')\SSO M^{\mathcal C}_{k_P}(k_V)$.

\medskip
\noindent\textbf{(i)$\Rightarrow$(ii).}
Assume (i): for all $k_V'<k_V$, $M(k_V')\SSO M(k_V)$. Fix any $\tau',\tau\in\mathcal C$ with
$\tau'\cx\tau$. We must show \eqref{eq:ID_pf}, i.e.\ $\int\Delta\phi_{k_V',k_V}\,d\tau'\ge\int\Delta\phi_{k_V',k_V}\,d\tau$.

By the richness assumption (R), there exist sequences $k_{V,n}\to k_V$ and $k_{V,n}'\to k_V'$
with $k_{V,n}'<k_{V,n}$ and selections
\(
\sigma_n\in M^{\mathcal C}_{k_P}(k_{V,n}),\; \sigma_n'\in M^{\mathcal C}_{k_P}(k_{V,n}')
\)
such that $\sigma_n\to\tau$ and $\sigma_n'\to\tau'$ in $\mathcal C$. For each $n$, since $k_{V,n}'<k_{V,n}$, assumption (i) gives
\(
M^{\mathcal C}_{k_P}(k_{V,n}')\ \SSO\ M^{\mathcal C}_{k_P}(k_{V,n}).
\)
Applying the definition of $\SSO$ to $\sigma_n'\in M^{\mathcal C}_{k_P}(k_{V,n}')$ and $\sigma_n\in M^{\mathcal C}_{k_P}(k_{V,n})$ yields
\(
\sigma_n'\vee\sigma_n\in M^{\mathcal C}_{k_P}(k_{V,n}'),\; \sigma_n'\wedge\sigma_n\in M^{\mathcal C}_{k_P}(k_{V,n}).
\)
Because $\mathcal C$ is a chain, either $\sigma_n'\cx\sigma_n$ or $\sigma_n\cx\sigma_n'$. If
$\sigma_n\cx\sigma_n'$, then $\sigma_n'\wedge\sigma_n=\sigma_n$ would belong to $M^{\mathcal C}_{k_P}(k_{V,n})$ and
$\sigma_n'\vee\sigma_n=\sigma_n'$ would belong to $M^{\mathcal C}_{k_P}(k_{V,n}')$, which forces $\sigma_n'\cx\sigma_n$
(or else the max/min membership would fail on a chain). Thus we may take
\(
\sigma_n'\cx \sigma_n \; \text{for all }n.
\)

Now use optimality at the parameter pair $(k_{V,n}',k_{V,n})$:
\(
U_{k_P}(\sigma_n,k_{V,n})\ge U_{k_P}(\sigma_n',k_{V,n}) \) and

\(U_{k_P}(\sigma_n',k_{V,n}')\ge U_{k_P}(\sigma_n,k_{V,n}').
\)
Adding and rearranging yields
\[
U_{k_P}(\sigma_n',k_{V,n}')-U_{k_P}(\sigma_n',k_{V,n})
\ge
U_{k_P}(\sigma_n,k_{V,n}')-U_{k_P}(\sigma_n,k_{V,n}),
\]
and by \eqref{eq:diff_id_pf} (with $(k_V',k_V)$ replaced by $(k_{V,n}',k_{V,n})$),
\begin{equation}\label{eq:ineq_n_pf}
\int \Delta\phi_{k_{V,n}',k_{V,n}}\,d\sigma_n'
\ \ge\
\int \Delta\phi_{k_{V,n}',k_{V,n}}\,d\sigma_n.
\end{equation}

Finally, let $n\to\infty$. Since $\sigma_n'\to\tau'$ and $\sigma_n\to\tau$ in $\mathcal C$ and
$k_{V,n}'\to k_V'$, $k_{V,n}\to k_V$, the integrals in \eqref{eq:ineq_n_pf} converge to
$\int \Delta\phi_{k_V',k_V}\,d\tau'$ and $\int \Delta\phi_{k_V',k_V}\,d\tau$, respectively (by the
continuity of $\phi_{k_V}(p)$ and dominated convergence on $[0,1]$). Hence
\(
\int \Delta\phi_{k_V',k_V}\,d\tau'
\ge
\int \Delta\phi_{k_V',k_V}\,d\tau,
\)
which is (ii). This completes the proof of equivalence.
\end{proof}

\begin{proof}[\textbf{\emph{Proof of Proposition \ref{thm:two_thresholds}.}}]
We will show the proof in steps:

\emph{Step 1. }
For fixed $(k_P,k_V)$, the outer problem is $\max_{\tau:\ \E_\tau[p]=p_0}\ \E_\tau[g_{k_P,k_V}(p)].$ In the binary-state case, standard concavification yields value $(\cav g_{k_P,k_V})(p_0)$ and an optimal $\tau^*$ supported on at most two posteriors. Moreover, pooling $\tau^*=\delta_{p_0}$ is optimal iff it achieves the value, i.e.\ $g_{k_P,k_V}(p_0)=(\cav g_{k_P,k_V})(p_0).$ Hence pooling is optimal iff the concavification gap $\Gamma(k_V):=(\cav g_{k_P,0})(p_0)-g_{k_P,0}(p_0)$ equals zero.

\emph{Step 2.}
Fix $k_V,k_V'\in[0,k]$ and $p\in[0,1]$. Let $a(p):=\Delta_u(p)$ and
$b(p):=C\big(q_{\min}(p)\big)$. By (A1), $b(p)\in[0,C(1)]$. Then
\(
\phi_{k_V}(p)-u(f)\cdot p=\max\{0,a(p)-k_V b(p)\},
\;
\phi_{k_V'}(p)-u(f)\cdot p=\max\{0,a(p)-k_V' b(p)\}.
\)
Using that $\max\{0,x\}$ is $1$--Lipschitz on $\R$, i.e.
\(
\big|\max\{0,x\}-\max\{0,y\}\big|\le |x-y|\,
\) for all \(x,y\in\R,\)
we obtain
\[
\begin{aligned}
|\phi_{k_V}(p)-\phi_{k_V'}(p)|
&=\big|\max\{0,a(p)-k_V b(p)\}-\max\{0,a(p)-k_V' b(p)\}\big|\\
&\le \big|(a(p)-k_V b(p))-(a(p)-k_V' b(p))\big|\\
&=|k_V-k_V'|\,b(p)
\le |k_V-k_V'|\,C(1).
\end{aligned}
\]
Taking the supremum over $p\in[0,1]$
yields
\(
\sup_{p\in[0,1]}|\phi_{k_V}(p)-\phi_{k_V'}(p)|
\le C(1)\,|k_V-k_V'|.
\)
Hence $k_V\mapsto \phi_{k_V}$ is Lipschitz in the uniform norm, and therefore uniformly
continuous.

\emph{Step 3.}
Since $\phi_{k_V}$ is uniformly continuous, by definition, $g_{k_P,k_V}$ depends continuously on $k_V$ in the uniform norm.
The concavification operator is continuous under uniform convergence on a compact interval:
if $g_n\to g$ uniformly on $[0,1]$, then $\cav g_n\to \cav g$ uniformly on $[0,1]$.
Therefore $(\cav g_{k_P,k_V})(p_0)$ is continuous in $k_V$, and so is $g_{k_P,k_V}(p_0)$ in $k_V$, hence $\Gamma(\cdot)$ is continuous.

\emph{Step 4.}
By (A2), $\Gamma(0)=0$. By (A3), $\Gamma(\bar k_V)>0$ for some $\bar k_V$, so the set
$\{k_V\ge 0:\Gamma(k_V)>0\}$ is nonempty, and $k_V^{\mathrm{ON}}(k_P)$ is well-defined and finite.
By definition of the infimum, if $k_V<k_V^{\mathrm{ON}}(k_P)$ then $\Gamma(k_V)=0$, hence pooling is optimal; this proves (i).

Let $\varepsilon>0$ be arbitrary. Suppose, for contradiction, that $\Gamma(k_V)=0$ for all $k_V\in\big(k_V^{\mathrm{ON}}(k_P),\,k_V^{\mathrm{ON}}(k_P)+\varepsilon\big)$.
Then there is \emph{no} $k_V$ in this interval with $\Gamma(k_V)>0$, which implies
\(
\inf\{k_V:\Gamma(k_V)>0\}\ \ge\ k_V^{\mathrm{ON}}(k_P)+\varepsilon,
\)
contradicting the definition of $k_V^{\mathrm{ON}}(k_P)$ as the infimum of the strict-positivity set.
Therefore, there must exist some
$k_V\in\big(k_V^{\mathrm{ON}}(k_P),\,k_V^{\mathrm{ON}}(k_P)+\varepsilon\big)$ with $\Gamma(k_V)>0$.
Given that $\Gamma(k_V)>0$ is equivalent to $(\cav g_{k_P,k_V})(p_0)>g_{k_P,k_V}(p_0)$, there exists an informative optimal information policy $\tau^*\neq \delta_{p_0}$ (Lemma~\ref{thm:concav}); this proves (ii).
(At $k_V=k_V^{\mathrm{ON}}(k_P)$, either pooling or informativeness may occur, depending on whether $\Gamma$ is zero or positive at the boundary.)

\emph{Step 5.}
Fix a posterior $p\in(\pi_L,\pi_H)$.
In the convex-hull cutoff formulation, to induce the agent to choose $g$ at $p$ the principal must choose at least $q_{\min}(p)$,
and the associated incremental net gain is $\Delta_u(p)-k_V C(q_{\min}(p))$.
Since any $q>q_{\min}(p)$ does not further change the action, the optimal management at $p$ is bang--bang:
\[
q^*(p;k_V)=
\begin{cases}
q_{\min}(p), & \Delta_u(p)>k_V C(q_{\min}(p)),\\
0, & \Delta_u(p)\le k_V C(q_{\min}(p)).
\end{cases}
\]
Therefore, if $k_V\ge \sup_{p\in(\pi_L,\pi_H)}\Delta_u(p)/C(q_{\min}(p))=:k_V^{NM}$, then
$\Delta_u(p)\le k_V C(q_{\min}(p))$ for all $p\in(\pi_L,\pi_H)$, implying $q^*(p;k_V)=0$ for all $p$.
Outside $(\pi_L,\pi_H)$, management is never used by construction (actions are unanimous), so (iii) follows.
Conversely, if $k_V<k_V^{NM}$, then by definition there exists some $p\in(\pi_L,\pi_H)$ such that
$\Delta_u(p)>k_V C(q_{\min}(p))$, hence $q^*(p;k_V)=q_{\min}(p)>0$, proving (iv).
\end{proof} 
\bibliographystyle{ecta} 
\bibliography{references}

\end{document}